\documentclass[oneside,english]{amsart}
\usepackage[T1]{fontenc}
\usepackage[latin9]{inputenc}
\usepackage{amsthm}
\usepackage{amssymb}
\usepackage{esint}

\makeatletter
\numberwithin{equation}{section}
\numberwithin{figure}{section}
\theoremstyle{plain}
\newtheorem{thm}{\protect\theoremname}[section]
  \theoremstyle{plain}
  \newtheorem{cor}[thm]{\protect\corollaryname}
  \theoremstyle{definition}
  \newtheorem{defn}[thm]{\protect\definitionname}
  \theoremstyle{remark}
  \newtheorem{rem}[thm]{\protect\remarkname}
  \theoremstyle{plain}
  \newtheorem{lem}[thm]{\protect\lemmaname}

\def\makebbb#1{
    \expandafter\gdef\csname#1\endcsname{
        \ensuremath{\Bbb{#1}}}
}\makebbb{R}\makebbb{N}\makebbb{Z}\makebbb{C}\makebbb{H}\makebbb{E}\makebbb{H}\makebbb{P}\makebbb{B}\makebbb{Q}\makebbb{E}

\usepackage{babel}

\usepackage{babel}

\makeatother

\usepackage{babel}

\makeatother

\usepackage{babel}
  \providecommand{\corollaryname}{Corollary}
  \providecommand{\definitionname}{Definition}
  \providecommand{\lemmaname}{Lemma}
  \providecommand{\remarkname}{Remark}
\providecommand{\theoremname}{Theorem}

\begin{document}

\title{On large deviations for Gibbs measures, mean energy and Gamma-convergence}

\author{Robert J. Berman}
\begin{abstract}
We consider the random point processes on a measure space $(X,\mu_{0})$
defined by the Gibbs measures associated to a given sequence of $N-$particle
Hamiltonians $H^{(N)}.$ Inspired by the method of Messer-Spohn for
proving concentration properties for the laws of the corresponding
empirical measures, we propose a number of hypotheses on $H^{(N)}$
which are quite general, but still strong enough to extend the approach
of Messer-Spohn. The hypotheses are formulated in terms of the asymptotics
of the corresponding mean energy functionals. We show that in many
situations the approach even yields a Large Deviation Principle (LDP)
for the corresponding laws. Connections to Gamma-convergence of (free)
energy type functionals at different levels are also explored. The
focus is on differences between positive and negative temperature
situations, motivated by applications to complex geometry. The results
yield, in particular, large deviation principles at positive as well
as negative temperatures for quite general classes of singular mean
field models with pair interactions, generalizing the 2D vortex model
and Coulomb gases. In a companion paper the results are illustrated
in the setting of Coulomb and Riesz type gases on a Riemannian manifold
$X,$ comparing with the complex geometric setting. 
\end{abstract}

\maketitle

\section{Introduction}

Let $X$ be a compact topological space endowed with a probability
measure $\mu_{0}.$ Given a sequence of symmetric functions $H^{(N)}$
on the $N-$fold products $X^{N},$ assumed measurable with respect
to $\mu_{0}^{\otimes N},$ the corresponding \emph{Gibbs measures}
at inverse temperature $\beta_{N}\in\R$ is defined as the following
sequence of symmetric probability measures on $X^{N}:$ 
\[
\mu_{\beta_{N}}^{(N)}:=e^{-\beta H^{(N)}}\mu_{0}/Z_{N,\beta},
\]
assuming that the partition function $Z_{N,\beta_{N}}$ is finite:
\[
Z_{N,\beta_{N}}:=\int_{X^{N}}e^{-\beta_{N}H^{(N)}}\mu_{0}^{\otimes N}<\infty
\]
We also assume that the following limit exists: 
\[
\beta:=\lim_{N\rightarrow\infty}\beta_{N}\in]-\infty,\infty]
\]
The ensemble $(X^{N},\mu_{\beta_{N}}^{(N)})$ (called the\emph{ canonical
ensemble}) defines a random point process with $N$ particles on $X$
which, from the point of view of statistical mechanics, models $N$
identical particles on $X$ interacting by the\emph{ Hamiltonian }(interaction
energy) $H^{(N)}$ in thermal equilibrium at \emph{inverse temperature}
$\beta_{N}.$ The corresponding \emph{empirical measure} is the random
measure 
\begin{equation}
\delta_{N}:\,\,X^{N}\rightarrow\mathcal{P}(X),\,\,\,(x_{1},\ldots,x_{N})\mapsto\delta_{N}(x_{1},\ldots,x_{N}):=\frac{1}{N}\sum_{i=1}^{N}\delta_{x_{i}}\label{eq:emp measure intro}
\end{equation}
 taking values in the space $\mathcal{P}(X)$ of of all probability
measures on $X$. A recurrent theme in statistical mechanics is to
study the large $N-$limit (i.e. the ``macroscopic limit'') of the
canonical ensemble through the large $N-$limit of the laws of $\delta_{N},$
i.e. through the sequence of probability measures 
\begin{equation}
\Gamma_{N}:=(\delta_{N})_{*}\mu_{\beta_{N}}^{(N)}\label{eq:law intro}
\end{equation}
on $\mathcal{P}(X).$ In many situations the laws $\Gamma_{N}$ can
be shown to concentrate, as $N\rightarrow\infty,$ at the subset of
$\mathcal{P}(X)$ consisting of the minima of a free energy type functional
$F_{\beta}$ on $\mathcal{P}(X);$ we will then say that \emph{``the
sequence $\Gamma_{N}$ has the concentration property''.} For example,
if the functional $F_{\beta}$ has a unique minimizer $\mu_{\beta}$
then it follows that the random measures $\delta_{N}$ converge in
law to a unique deterministic measure $\mu_{\beta}.$ A stronger exponential
notion of concentration, with an explicit speed and rate functional,
is offered by the theory of large deviations, by demanding that the
laws $\Gamma_{N}$ satisfy a\emph{ Large Deviation Principle (LDP)}
with \emph{speed} $r_{N}$ and a \emph{rate functional} $F,$ symbolically
expressed as 
\[
\Gamma_{N}(\mu)\sim e^{-r_{N}F(\mu)},\,\,N\rightarrow\infty
\]
The present paper is inspired by the method introduced by Messer-Spohn
\cite{m-s} to establish the concentration property of the laws of
the empirical measures $\delta_{N}$ using the Gibbs variational principle
combined with properties of the mean energy of the system. In the
original approach in \cite{m-s} $H^{(N)}$ was assumed to be the
\emph{mean field Hamiltonian} corresponding to a continuous pair interaction
potential $W(x,y):$ 
\[
H^{(N)}(x_{1},\ldots,x_{N}):=\frac{1}{N}\sum_{1\leq i,j\leq N}W(x_{i},x_{j}),
\]
 where $\beta_{N}=\beta\in]0,\infty[$ (this is a mean field interaction
in the sense that each particle $x_{i}$ is exposed to the average
of the pair interactions $W(x_{i},x_{j})$ of all $N$ particles,
including the self-interaction). But the approach has also been extended
to handle some situations where $W(x,y)$ is allowed to be singular
\cite{clmp,k2,ki-s}, as in Onsager's vortex model for 2D turbulence
\cite{o}, where $W(x,y)=-\log|x-y|.$ The corresponding mean field
Hamiltonian is then ``renormalized'' by removing the self-interaction
terms in order to make sure that $H^{(N)}$ is generically finite
on $X^{N}.$ The aim of the present note is to
\begin{itemize}
\item Propose a number of quite general hypothesis on $H^{(N)},$ formulated
in terms of the corresponding mean energy functional $E^{(N)}$on
$\mathcal{P}(X^{N}),$ which are strong enough to extend the approach
of Messer-Spohn.
\item Show that the approach also yields the stronger exponential concentration
property in the sense of a LDP, almost ``for free'', in several
situations
\item Explore some relations to the notion of Gamma-convergence of functionals:
first by reformulating the approach of Messer-Spohn in terms of Gamma-convergence
of the induced free energy functionals $F_{\beta_{N}}^{(N)}$ on $\mathcal{P}(\mathcal{P}(X))$
and then by deducing a Gamma-convergence result for the sequence $H^{(N)}/N$
on $X^{N}.$ 
\end{itemize}
The main motivation comes from the probabilistic approach to the construction
of Kähler-Einstein metrics on a complex algebraic manifold $X$, introduced
in\cite{berm8,berm8b}. In that situation the corresponding Hamiltonians
$H^{(N)}$ are highly non-linear and singular (and not of the simpler
mean field type appearing in formula\ref{eq:def of finite order intro}
which is used in the statistical mechanical approach to conformal
geometry introduced in \cite{k2-1}). But still, as shown in \cite{berm8},
building on \cite{b-b}, the sequence $H^{(N)}/N$ Gamma-converges
towards a certain energy type functional $E(\mu)$ on $\mathcal{P}(X).$
Exploiting superharmonicity properties of $H^{(N)}/N$ the corresponding
LDP is then established at a any positive inverse temperature $\beta$
in \cite{berm8}, producing Kähler-Einstein metrics with negative
Ricci curvature in the large $N-$limit. The approach in \cite{berm8}
bypasses the problem of the existence of the macroscopic mean energy
(hypothesis H1 below), which is open in the complex geometric setting.
On the other hand, as discussed in \cite{berm8}, extending the LDP
in \cite{berm8} to negative $\beta,$ which is needed to produce
Kähler-Einstein metrics with \emph{positive }Ricci curvature, necessitates
the existence of the macroscopic mean energy. This is the motivation
behind Theorem \ref{thm:h1 and h4 neg beta intro} below which shows
that, conversely, hypothesis H1 together with the additional hypotheses
H4, implies a LDP for appropriate negative $\beta.$ The hypothesis
H4 is inspired by the energy-entropy compactness results in \cite{bbgez},
which can be viewed as the macroscopic analog of H4 in the complex
geometric setting. Incidentally, in the lowest-dimenaional case when
$X$ is a Riemann surface the probabilistic setting in \cite{berm8,berm8b}
is essentially equivalent to a mean field model with a logarithmic
pair interaction, which is thus similar to Onsager's vortex model
for 2D turbulence \cite{o}. In the latter situation the corresponding
concentration properties were established in \cite{clmp,k2}, for
any $\beta$ above the critical negative temperature (and the LDP
was established using a different method in \cite{bo-g}). This in
line with Onsager's prediction of the existence of macroscopic negative
temperature states. 

Another motivation for the present note comes from random matrix theory
(or more generally Coulomb gases) which can be viewed as a vortex
type model with $\beta_{N}\sim N$ (and in particular $\beta=\infty.$)

The corresponding concentration property was established in \cite{ki-s},
using the method of Messer-Spohn as in \cite{clmp,k2}. Here we observe
that, with a simple modification, the concentration property can be
upgraded to a LDP (Corollary \ref{cor:cor mean field pos beta intro}).
In particular, this allows one to dispense with the technical assumption
that $\beta_{N}\gg\log N$ which is needed in the approach in \cite{ben-g,c-g-z,be-z,se}),
where moreover some regularity properties of the corresponding pair
interaction $W(x,y)$ away from the diagonal are required (see Section
\ref{sub:Relations-to-Gamma intro}). \footnote{The Hamiltonians in the random matrix and Coulomb gas literature are
usually scaled in a different way so that our zero-temperature $(\beta=\infty)$
corresponds to a fixed inverse temperature.}

Yet another motivation comes from approximation and sampling theory
and, in particular, the problem of finding nearly minimal configurations
for a given energy type interaction on a Riemannian manifold, in the
spirit of \cite{h-s,s-k}. 

Let us also point out that that the restriction that $X$ be compact
can be removed if suitable growth-assumptions of $H^{(N)}$ at infinity
are made, as in the settings in $\R^{n}$ considered in \cite{ben-g,c-g-z,be-z,se,ki-s,d-l-r}
(using appropriate tightness estimates). But in order to (hopefully)
convey the conceptual simplicity of the arguments we stick with a
compact $X.$

\subsection{Hypotheses\label{sub:Hypotheses}}

We may as well assume that $X$ coincides with the support of $\mu_{0}.$
In the following $\mu^{(N)}$ will denote a symmetric probability
measure on $X^{N}$ and $Y:=\mathcal{P}(X).$ We recall that the \emph{mean
(microscopic) energy }of $\mu^{(N)},$ in the usual sense of statistical
mechanics, is defined by 

\[
E^{(N)}(\mu^{(N)}):=\frac{1}{N}\int_{X^{N}}H^{(N)}\mu^{(N)}
\]
We introduce the following hypotheses: 
\begin{itemize}
\item (H1) (``existence of a macroscopic mean energy $E(\mu)"$): There
exists a functional $E(\mu)$ on $\mathcal{P}(X)$ such that for any
$\mu$ in $\mathcal{P}(X)$ satisfying $E(\mu)<\infty$
\[
\lim_{N\rightarrow\infty}E^{(N)}(\mu^{\otimes N})=E(\mu)
\]
Moreover, $E(\mu_{0})<\infty$
\item (H2) (``lower bound on the mean energy'') For any sequence of $\mu^{(N)}$
such that $\Gamma_{N}:=(\delta_{N})_{*}\mu^{(N)}\rightarrow\Gamma$
weakly in $\mathcal{P}(Y)$ we have 
\[
\liminf_{N\rightarrow\infty}E^{(N)}(\mu^{(N)})\geq E(\Gamma):=\int_{Y}E(\mu)\Gamma(\mu)
\]

\item (H3) ``Approximation property'': For any $\mu$ such that $E(\mu)<\infty$
there exists a sequence $\mu_{j}$ converging weakly to $\mu$ such
that $\mu_{j}$ is absolutely continuous with respect to $\mu_{0}$
and satisfies $E(\mu_{j})\rightarrow E(\mu).$
\item (H4) (``mean energy/entropy compactness'') If 
\[
E^{(N)}(\mu^{(N)})\leq C,\,\,\,D^{(N)}(\mu^{(N)})\leq C,
\]
 where $D^{(N)}(\mu^{(N)})$ is the mean entropy, then the following
convergence holds, after perhaps replacing $\mu^{(N)}$ by a subsequence
such that $\Gamma_{N}:=(\delta_{N})_{*}\mu^{(N)}\rightarrow\Gamma$
weakly in $\mathcal{P}(Y):$ 
\[
\lim_{N\rightarrow\infty}E^{(N)}(\mu^{(N)})=\int_{Y}E(\mu)\Gamma(\mu)
\]

\end{itemize}
The first hypothesis will be assumed through out the paper. The second
and third one will appear naturally in positive and vanishing temperature
respectively, while the fourth one turns out to be useful in some
case of negative temperature. However, it may very well be that the
hypothesis H4 needs to be weakened a bit in order to increase its
scope. For example, in the proof of the large $N-$ concentration
properties one only needs to assume that H4 holds when $\mu^{(N)}$
is the Gibbs measure corresponding to $H^{(N)}$ (see Remark \ref{rem:only concentration}). 

Of course, the sign of the temperature may be switched by replacing
$H^{(N)}$ with $-H^{(N)},$ but the point is that, in practice, we
will consider settings where the sign of $H^{(N)}$ is fixed by the
requirement that $H^{(N)}$ be bounded from below (which essentially
means that the system is assumed to be stable at zero temperature).

\subsection{Large deviation results\label{sub:Large-deviation-results}}

We start with the simpler setting of positive temperature:
\begin{thm}
\label{thm:h1 and h2 pos beta intro}Suppose that hypotheses H1 and
H2 hold and let $\beta_{N}$ be a sequence of positive numbers tending
to $\beta\in]0,\infty].$ Then the measures $(\delta_{N})_{*}(e^{-\beta_{N}H^{(N)}}\mu_{0}^{\otimes N})$
on $\mathcal{P}(X)$ satisfy, as $N\rightarrow\infty,$ a large deviation
principle (LDP) with \emph{speed} $\beta_{N}N$ and \emph{rate functional}
\begin{equation}
F_{\beta}(\mu)=E(\mu)+\frac{1}{\beta}D_{\mu_{0}}(\mu)\label{eq:free energy func theorem gibbs intro}
\end{equation}
Equivalently, the LDP holds for the corresponding Gibbs measures with
$F_{\beta}$ replaced by $F_{\beta}-\inf F_{\beta}.$ Under the additional
hypothesis H3 the result also holds when $\beta=\infty$ 
\end{thm}
The previous theorem in particular applies to the following ``finite
order'' Hamiltonians of mean field type. Given symmetric functions
$W_{m}$ on $X^{m}$ for $m\leq M$ set 
\begin{equation}
H^{(N)}(x_{1},...x_{N}):=\frac{1}{N^{(m-1)}}\sum_{m=1}^{M}\sum_{I}W_{m}(x_{i_{1}},...,x_{i_{m}}),\label{eq:def of finite order intro}
\end{equation}
 where the inner sum rums over all multi indices $I=(i_{1_{i}},...,i_{m})$
of length $m$ and with the property that no two indices of $I$ coincide.
Then it is easy to verify H1 and H2 above with 
\[
E(\mu):=\sum_{m=1}^{M}\int_{X^{m}}W_{m}\mu^{\otimes m}
\]
The main case of interest is when $M=2$ and $H^{(N)}$ is a sum of
pair-interactions $W(x_{i},x_{j}),$ scaled by $1/N.$ But since it
will require no extra effort in the proofs we consider the more general
``finite order setting''.
\begin{cor}
\label{cor:cor mean field pos beta intro}Let $W(x_{1},..x_{m})$
be a symmetric lower semi-continuous function on $X^{m}$ and $\beta_{N}$
a sequence of positive numbers tending to $\beta\in]0,\infty[,$ Assume
that the Gibbs measures $\mu_{\beta_{N}}^{(N)}$ of the corresponding
mean field Hamiltonians are well-defined probability measures. Then
the laws $(\delta_{N})_{*}\mu_{\beta}^{(N)}$ satisfy a LDP with speed
$\beta_{N}N$ and rate functional $F_{\beta},$ with 
\[
E(\mu):=\int_{X^{m}}W\mu^{\otimes m}
\]
The corresponding result also holds for $\beta=\infty$ if the regularization
hypothesis H3 holds.
\end{cor}
In the Euclidean setting and with $M=2$ the previous corollary was
established very recently in\cite{d-l-r} using different methods
(based on control theory).

We next turn to the case of negative temperature.
\begin{thm}
\label{thm:h1 and h4 neg beta intro}Suppose that hypothesis H1 and
H4 hold and fix a negative number $\beta_{0}.$ Then the following
is equivalent: 
\begin{itemize}
\item For any $\beta>\beta_{0}$ we have $Z_{N,\beta}\leq C_{\beta}^{N}$ 
\item For any $\beta>\beta_{0}$ the measures $(\delta_{N})_{*}\left(e^{-\beta H^{(N_{k})}}\mu_{0}^{\otimes N_{k}}\right)$
on $\mathcal{P}(X)$ satisfy a LDP with speed $N$ and rate functional
\[
\beta F_{\beta}(\mu)=\beta E(\mu)+D_{\mu_{0}}(\mu)
\]

\end{itemize}
\end{thm}
\begin{cor}
\label{cor:mean field neg beta intro}Let $W(x,y)$ be a symmetric
measurable function on $(X^{2},\mu_{0}^{\otimes2})$ such that for
any $\beta\in\R$ 
\begin{equation}
\int_{X^{2}}e^{-\beta W}\mu_{0}^{\otimes2}<\infty\label{eq:integr cond in cor}
\end{equation}
Then, for any $\beta\in]-\infty,\infty[,$ the Gibbs measures $\mu_{\beta}^{(N)}$
of the mean field Hamiltonians corresponding to $W$ satisfy a LDP
with speed $N$ and rate functional $\beta F_{\beta},$ with 
\[
E(\mu):=\int_{X^{2}}W\mu^{\otimes2}
\]

\end{cor}
For example, the previous corollary applies when $(X,\mu_{0})$ is
a domain in $\R^{D}$ endowed with the Lebesgue measure and $W$ is
any symmetric function with a gradient in $L_{loc}^{D}(\R^{D})$ (then
the exponential integrability condition follows from Trudinger's inequality).
The key observation in the proof of Cor \ref{cor:mean field neg beta intro}
is that the first point in Theorem \ref{thm:h1 and h4 neg beta intro}
always implies, ``for free'', a uniform estimate in the Orlitz (Zygmund)
space $L^{1}\mbox{Log}L^{1},$ so that some general Orlitz space duality
results \cite{r-r,le} can be exploited in order to verify the hypothesis
H4. It seems natural to ask if the previous corollary can be generalized
to the case when the integrability condition \ref{eq:integr cond in cor}
is only assumed to hold for $\beta>\beta_{0},$ for some (finite)
negative number $\beta_{0}?$ The following theorem gives an affirmative
answer if one strengthens the integrability condition a bit:
\begin{thm}
\label{thm:negive beta uniform integr}Let $X$ be a compact metric
space and $W$ a lower semi-continuous symmetric measurable function
on $X^{2}$ and $\beta_{0}$ a negative number such that
\[
\sup_{x\in X}\int_{X}e^{-\beta_{0}W(x,y)}\mu_{0}(y)<\infty
\]
Then, for any $\beta>\beta_{0}$ the Gibbs measures $\mu_{\beta}^{(N)}$
satisfy an LDP as in the previous corollary.
\end{thm}
Specialized to the logarithmic case of the vortex model, i.e. to the
case
\[
W(x,y)=-\log|x-y|,
\]
 the previous theorem recovers the LDP in \cite{bo-g} with a new
proof. The proof follows closely the corresponding (weaker) concentration
result for the vortex model originally established in \cite{clmp,k2}.
The new observation is that with a little twist the argument in \cite{clmp,k2}
can be supplemented to give the LDP in question.

\subsection{\label{sub:Relations-to-Gamma intro}Relations to Gamma-convergence
at different levels}

The proofs of the LDPs above are based on the Gamma-convergence of
the corresponding free energy functionals $F_{\beta_{N}}^{(N)}$ when
viewed as functionals on the space $\mathcal{P}(\mathcal{P}(X))$
(a similar approach is used in the dynamic setting considered in \cite{b-o}
where the assumptions H1 and H2 also appear naturally). Incidentally,
as observed in the following corollary the LDPs then imply the Gamma-convergence
of  the scaled Hamiltonians $H^{(N)}/N$ when viewed as functionals
on $\mathcal{P}(X).$
\begin{cor}
\label{cor:h1 h2 h3 gives gamma intro}Suppose that the Hamiltonians
$H^{(N)}$ satisfy H1 and H2 and H3 for some measure $\mu_{0}.$ Then
$H^{(N)}/N$ Gamma-converges towards $E(\mu)$ on $\mathcal{P}(X).$
In particular, this applies to the finite order mean field Hamiltonians
(assuming the approximation property H3).
\end{cor}
The previous result generalizes the Gamma-convergence result in \cite[Prop 2.8, Remark 2.19]{se}
for the mean field Hamiltonians corresponding to pair interactions
$W(x,y)$ (on Euclidean domains), where it was assumed that $w(x,y)$
essentially only blows up along the diagonal (similar result also
appear implicitly in \cite{ben-g,c-g-z,be-z}). The proofs in \cite{se,ben-g,c-g-z,be-z})
are based on some rather intricate combinatorial constructions, involving
small cubes. On the other hand, the latter results yield the stronger
result that any measure $\mu$ with a positive continuous density
admits a recovery sequences $x^{(N)}$ such that 
\[
\limsup_{N\rightarrow\infty}\sup_{B_{\epsilon_{N}}}E_{N}\rightarrow E(\mu),
\]
 where $B_{\epsilon_{N}}$ denotes the $L^{\infty}-$ball centered
at $x^{(N)}$ with radius $\epsilon_{N}$ of the order $1/N^{1/D},$
for $D=\dim X.$ In turn, as shown in \cite{se}, the latter stronger
form of Gamma-convergence implies the LDP for the corresponding Gibbs
measures when $\beta_{N}\gg(\log N)$ (which is needed to make sure
that $\beta_{N}^{-1}N^{-1}\log\int_{B_{\epsilon_{N}}}dV^{\otimes}\rightarrow0).$ 

Relations between Gamma-convergence and large deviation principles
have also been previously studied in \cite{ma} but from a rather
different perspective (see also \cite{bo} for some related results,
which in particular allows one to dispense with the assumption H3
in \ref{cor:h1 h2 h3 gives gamma intro}).

\subsection{Applications to the Coulomb gas on a Riemannian manifold}

In the companion paper \cite{b2} the general large deviation results
above are illustrated and further developed for Coulomb and Riesz
type gases on a compact $D-$dimensional Riemannian manifold $(X,g)$
(and more generall for suitable compact subsets $K\subset X).$ Here
we will only state the corresponding LDP for the Coulomb gas on $(X,g)$
defined as follows. Let $W(x,y)$ be $1/2$ times the integral kernel
of the inverse of the positive Laplacian $-\Delta$ on the space of
all functions in $L^{2}(X,dV_{g})$ with mean zero. As is classical,
$W$ is symmetric and smooth away from the diagonal and close to the
diagonal it admits the following asymptotics when $D>2:$ 
\[
W(x,y)=\frac{\Gamma(D/2-1)}{(4\pi)^{D/2}}\frac{1}{d(x,y)^{(D-2)}}(1+O(1)),\,\,\,D>2,
\]
where the leading constant is expressed in terms of the classical
$\Gamma-$function. Moreover, when $D=2$ 
\[
W(x,y)=-\frac{1}{(4\pi)}\log d^{2}(x,y)+O(1),\,\,D=2
\]
In particular, $W$ is lsc and in $L^{1}(X)$ on $X\times X.$ Given
a probality measure $\mu_{0}$ on $X$ the Coulomb gas on $(X,g,\mu_{0})$
at inverse temperature $\beta_{N}$ is defined by the Gibbs measures
corresponding to $(\mu_{0},H^{(N)},\beta_{N})$ where $H^{(N)}$ is
the mean field Hamiltonian corresponding to the pair interaction $W(x,y).$
In this setting Cor \ref{cor:cor mean field pos beta intro} and Theorem
\ref{thm:negive beta uniform integr} yields, as shown in \cite{b2},
the following LDP for the laws of the empirical measures of the Coulomb
gas, formulated in terms of the potential theoretic properties of
the measure $\mu_{0}:$
\begin{thm}
\label{thm:LDP coul}Let $(X,g)$ be a compact Riemannian manifold
and consider the Coulomb gas at inverse temperature $\beta_{N}$ on
$(X,g,\mu_{0}).$ 
\begin{itemize}
\item When $\beta\in]0,\infty[$ the LDP holds if the measure $\mu_{0}$
is non-polar 
\item When $\beta=\infty$ the LDP holds if $\mu_{0}$ is non-polar and
$\mu_{0}$ is determining for its support $K.$ 
\item When $D=2$ and $\beta<\infty$ the LDP holds when $\beta>-4\pi d(\mu_{0}),$
where $d(\mu_{0})\in[0,\infty[$ is the sup över all $t>0$ such that
there exists a positive constant $C$ (depending on $t)$ such that
\[
\mu_{0}(B_{R}(x))\leq CR^{t}
\]
\textup{as $R\rightarrow0,$ for any Riemannian ball $B_{R}(x)$ of
radius $R$ centered at a given point $x$ in $X.$ }
\end{itemize}
\end{thm}
We briefly recall that a compact subset $K\subset X$ is\emph{ polar}
if it is locally contained in the $-\infty-$set of a local subharmonic
function (or equivalently, if $K$ has vanishing capacity). Accordingly,
a measure $\mu_{0}$ is said to be\emph{ non-polar} if it does not
charge any polar set. The notion of a\emph{ determining }measure $\mu_{0}$
appearing in the second point above means that for any given $u\in C^{0}(X)$
\[
\left\Vert e^{\varphi-u}\right\Vert _{L^{\infty}(K,\mu_{0})}=\sup_{K}e^{\varphi-u},
\]
 for any quasi-subharmonic function $\varphi$ on $X,$ i.e. $\varphi$
is strongly usc and satisfies $\Delta\varphi\geq-1.$ This notion
is closely related to the notion of measures satisfying a \emph{Bernstein-Markov
property} in pluripotential theory \cite{blwp,b-b-w} and measures
with \emph{regular asymptotic behaviour} in the theory of planar orthogonal
polynomials \cite{s-t}. For example, $\mu_{0}$ can be taken to be
the $D-$dimensional Hausdorff measure on a Lipschitz domain $K\subset(X,g)$
or the $(D-1)-$dimensional Hausdorff measure on a Lipschitz hypersurface
in $(X,g).$ The point is that the assumption that $\mu_{0}$ is non-polar
and determining implies that the hypothesis H3 is satisfied, as shown
in \cite{b2} (an alternative proof of the LDP in the case when $\beta=\infty$
can also be given using the approach in the complex geometric setting
in \cite{berm 1 komma 5}, based on \cite{b-b,b-h}). Finally, we
recall that measures satisfying $d(\mu_{0})>0,$ as in the third point
above, are sometimes called \emph{Frostman measures }in the classical
litterature (for example, the property in question holds with $d(\mu_{0})=d$
when $\mu_{0}$ is the $d-$ dimensional Hausdorff measure of a compact
subset $K$ of $X$ of Hausdorff dimension $d).$ 

More generally, an LDP as in the previous theorem is obtained in \cite{b2},
when $W(x,y)$ is taken as the integral kernel of the inverse of $(-\Delta)^{p}$
and the (possible fractional) power $p$ is in $]0,D/2]$ (or even
more generally: when $(-\Delta)^{p}$ is replaced by a suitable pseudodifferential
operator of order at most $D).$ Then the last point in the previous
theorem holds in the critical case $p=D/2.$ However, the LDP for
$\beta=\infty$ appears to be rather subtle in the general setting
and is only shown to hold when $\mu_{0}$ is a volume form (or comparable
to a volume form), except when $p\leq2$ where it applies to measures
$\mu_{0}$ which are determining in a suitable sense. 

Let us also point out that in the Euclidean setting of the Coulomb
and Riesz gases in $\R^{n},$ with $\mu_{0}$ given by the Euclidean
volume form and $\beta_{N}$ of the order $N,$ a refined ``microscopic''
large deviation principle ``at the level of processes'' is obtained
in \cite{l-s}. Such large deviation principles are beyond the scope
of the present paper and seem to require different methods - the point
here is rather to allow the measure $\mu_{0}$ to be very singular
(and the inverse temperature to be negative, in some cases).

\subsubsection*{Acknowledgment}

It is a pleasure to thank Sebastien Boucksom, Vincent Guedj, Philippe
Eyssidieu and Ahmed Zeriahi for the stimulating collaboration \cite{bbgez}.
Also thanks to the editors Doug Hardin, Edward Saff and Sylvia Serfaty
for the invitation to contribute to the special issue of the Journal
of Constructive Approximation on the theme `Approximation and statistical
physics\textquoteright , which prompted the present paper.

\section{\label{sec:Proofs-of-the}Proofs of the large devations results}

\subsection{General notation}

Given a compact topological space $X$ we will denote by $C^{0}(X)$
the space of all continuous functions $u$ on $X,$ equipped with
the sup-norm and by $\mathcal{M}(X)$ the space of all signed (Borel)
measures on $X.$ The subset of $\mathcal{M}(X)$ consisting of all
probability measures will be denoted by $\mathcal{P}(X).$ We endow
$\mathcal{M}(X)$ with the weak topology, i.e. $\mu_{j}$ is said
to converge to $\mu$ weakly in $\mathcal{M}(X)$ if 
\[
\left\langle \mu_{j},u_{j}\right\rangle \rightarrow\left\langle \mu,u\right\rangle :=\int_{X}u\mu
\]
 for any continuous function $u$ on $X,$ i.e. for any $u\in C^{0}(X)$
(in other words, the weak topology is the weak{*}-topology when $\mathcal{M}(X)$
is identified with the topological dual of $C^{0}(X)).$ Since $X$
is compact so is $\mathcal{P}(X).$ Given a lower semi-continuous
function $F$ on 
\[
Y:=\mathcal{P}(X)
\]
 we will, abusing notation slightly, also write $F$ for the induced
linear lower semi-continuous functional on $\mathcal{P}(Y):$ 
\[
F(\Gamma):=\int_{\mathcal{P}(Y)}F(\mu)\Gamma(\mu)
\]
Equivalently, under the natural embedding $\mu\mapsto\delta_{\mu}$
of $Y$ into $\mathcal{P}(Y)$ the function $F(\Gamma)$ is the unique
lower semi-continuous affine extension of $F$ to $\mathcal{P}(Y).$ 

We will denote by $S_{N}$ the permutation group acting on $X^{N}$
and by $\mathcal{P}(X^{N})^{S_{N}}$ the space of symmetric measures
$\mu_{N}$ (i.e. $S_{N}-$invariant) on $X^{N}.$ Also note that,
following standard practice, we will denote by $C$ a generic constant
whose value may change from line to line.

\subsubsection{Entropy}

We will write $D(\nu_{1},\nu_{2})$ for the \emph{relative entropy
}(also called the \emph{Kullback\textendash Leibler divergence} in
information theory) of two measures $\nu_{1}$ and $\nu_{2}$ on a
topological space $Z:$ if $\nu_{1}$ is absolutely continuous with
respect to $\nu_{2},$ i.e. $\nu_{1}=f\nu_{2},$ one defines 
\[
D(\nu_{1},\nu_{2}):=\int_{Y}\log(\nu_{1}/\nu_{2})\nu_{1}
\]
 and otherwise one declares that $D(\mu):=\infty.$ Note the sign
convention used: $D$ is minus the \emph{physical} entropy. In our
setting the space $Z$ will always be of the form $X^{N}$ and we
will then take the reference measure $\nu_{2}=\mu_{0}^{\otimes N}$
and write $D(\cdot):=D(\cdot,\mu_{0}^{\otimes N}).$ It will also
be convenient to define the\emph{ mean entropy} of a probability measure
$\mu_{N}$ on $X^{N}$ (i.e. $\mu_{N}\in\mathcal{M}_{1}(X^{N}))$
as 
\[
D^{(N)}(\mu_{N}):=\frac{1}{N}D(\mu_{N},\mu_{0}^{\otimes N}).
\]
Then it follows directly that 
\begin{equation}
D^{(N)}(\mu^{\otimes N})=D(\mu).\label{eq:entropi for product}
\end{equation}
 Moreover, denoting by $\left(\mu_{N}\right)_{j}$ the $j$ th marginal
$\mu_{N}$ (which defines a probability measure on $X^{j})$ 
\begin{equation}
D^{(N)}(\mu_{N})\geq D^{(j)}(\left(\mu_{N}\right)_{j}),\label{eq:ineq for entropy}
\end{equation}
 as follows from the concavity of the function $t\mapsto\log t$ on
$\R_{+}$ (see for example \cite{k2}).

\subsection{Preliminaries}

\subsubsection{Large deviation principles}

Let us start by recalling the general definition of a Large Deviation
Principle (LDP) for a sequence of measures.
\begin{defn}
\label{def:large dev}Let $Y$ be a Polish space, i.e. a complete
separable metric space.

$(i)$ A function $I:\mathcal{\,Y}\rightarrow]-\infty,\infty]$ is
a \emph{rate function} if it is lower semi-continuous. It is a \emph{good}
\emph{rate function} if it is also proper.

$(ii)$ A sequence $\Gamma_{N}$ of measures on $Y$ satisfies a \emph{large
deviation principle} with \emph{speed} $r_{N}$ and \emph{rate function}
$I$ if

\[
\limsup_{N\rightarrow\infty}\frac{1}{r_{N}}\log\Gamma_{N}(\mathcal{F})\leq-\inf_{\mu\in\mathcal{F}}I
\]
 for any closed subset $\mathcal{F}$ of $Y$ and 
\[
\liminf_{N\rightarrow\infty}\frac{1}{r_{N}}\log\Gamma_{N}(\mathcal{G})\geq-\inf_{\mu\in G}I(\mu)
\]
 for any open subset $\mathcal{G}$ of $Y.$ \end{defn}
\begin{rem}
The LDP is said to be \emph{weak }if the upper bound is only assumed
to hold when $\mathcal{F}$ is compact. Anyway, we will only consider
the case when $Y$ is compact and hence the notion of a weak LDP and
an LDP then coincide (and moreover any rate functional is automatically
good).\end{rem}
\begin{lem}
(Bryc). Let $Y$ be a compact Polish space. Suppose that there exists
a function $f$ on $C^{0}(Y)$ such that for any $\Phi\in C^{0}(Y)$
\[
f_{N}(\Phi):=\frac{1}{r_{N}}\log\int e^{r_{N}\Phi}\Gamma_{N}\rightarrow f(\Phi)
\]
Then $\Gamma_{N}$ satisfies a LDP with speed $r_{N}$ and rate functional
\[
I(\mu)=\sup_{\Phi\in C^{0}(Y)}\left(\Phi(\mu)-f(\Phi)\right)
\]
(by Varadhan's lemma the converse also holds). 
\end{lem}
We also have the following simple lemma which allows one to apply
Bryc's lemma to the non-normalized measures $(\delta_{N})_{*}e^{-\beta H^{(N)}}\mu_{0}^{\otimes N}:$ 
\begin{lem}
\label{lem:ldp for nonnormalized meas}Assume that $\left|\log Z_{N,\beta}\right|\leq CN.$
Then the measures $(\delta_{N})_{*}e^{-\beta H^{(N)}}\mu_{0}^{\otimes N}$
satisfy the asymptotics in Bryc's lemma with rate functional $\tilde{I}(\mu)$
and speed $N$ iff the corresponding probability measures $(\delta_{N})_{*}\mu_{\beta}^{(N)}$
on $\mathcal{P}(X)$ satisfy an LDP at speed $N$ with rate functional
$I:=\tilde{I}-C_{\beta},$ where $C_{\beta}:=\inf_{\mathcal{\mu\in}\mathcal{P}(X)}I(\mu).$
\end{lem}

\subsubsection{Gamma-convergence}

We recall that a sequence of functions $f_{j}$ on a topological space
\emph{$\mathcal{X}$ }is\emph{ }said to\emph{ Gamma-converge }to a
function $f$ on $\mathcal{X}$ if 
\begin{equation}
\begin{array}{ccc}
x_{j}\rightarrow x\,\mbox{in\,}\mathcal{X} & \implies & \liminf_{j\rightarrow\infty}f_{j}(x_{j})\geq f(x)\\
\forall x & \exists x_{j}\rightarrow x\,\mbox{in\,}\mathcal{X}: & \lim_{j\rightarrow\infty}f_{j}(x_{j})=f(x)
\end{array}\label{eq:def of gamma conv}
\end{equation}
(such a sequence $x_{j}$ is called a\emph{ recovery sequence}); see
\cite{bra}. More generally, given a subset $\mathcal{S}\Subset\mathcal{X}$
we will say that $f_{j}$ \emph{Gamma-converge to $f$ relative to
$\mathcal{S}$} if the existence of a recovery sequence in $\mathcal{X}$
is only demanded when $x\in\mathcal{S}.$ 
\begin{lem}
\label{lem:gamma is lsc}Assume that $f_{j}$ Gamma-converges to $f$
relative to $\mathcal{S}\subset\mathcal{X}.$ Then $f_{|\mathcal{S}}$
is lower semi-continuous.\end{lem}
\begin{proof}
Consider a sequence $s_{i}\rightarrow s$ in $S.$ For each $s_{i}$
we take a recovery sequence $x_{i}^{(j)}$ in $X$ converging so $s_{i}.$
By a diagonal argument we get a sequence $x_{i}$ in $\mathcal{X}$
converging to $s$ such that $f(s_{i})=f_{i}(x_{i})+o(1)$ and hence
$f(s_{i})\geq f(s)+o(1),$ as desired.\end{proof}
\begin{lem}
\label{lem:conv of inf}Let $\mathcal{X}$ be a compact topological
space and assume that $f_{j}$ \emph{Gamma-converges }to $f$ relative
to a set $\mathcal{S}$ containing all minima of $f.$ Then 
\[
\lim_{j\rightarrow\infty}\inf_{\mathcal{X}}f_{j}=\inf_{\mathcal{X}}f
\]
\end{lem}
\begin{proof}
Given $s\in\mathcal{S}$ we take a recovery sequence $x_{n}$ and
observe that 
\[
f(s)\geq f_{n}(x_{n})+o(1)\geq\inf f_{n}+o(1)=f_{n}(y_{n})+o(1)\geq f(y)+o(1),
\]
 for some $y\in\mathcal{X},$ by the compactness and the assumption
of Gamma-convergence. In particular, when $s$ realizes the minimum
of $f$ so does $y$ and hence equalities must hold above, which concludes
the proof. 
\end{proof}

\subsubsection{Legendre-Fenchel transforms}

Let $f$ be a function on a topological vector space $V.$ Then its
Legendre-Fenchel transform is defined as following convex lower semi-continuous
function $f^{*}$ on the topological dual $V^{*}$ 
\[
f^{*}(w):=\sup_{v\in V}\left\langle v,w\right\rangle -f(v)
\]
in terms of the canonical pairing between $V$ and $V^{*}.$ In the
present setting we will take $V=C^{0}(X)$ and $V^{*}=\mathcal{M}(X),$
the space of all signed Borel measures on $X.$ Then $f^{**}=f$ for
any lower semi-continuous convex function (by standard duality in
locally convex topological vector spaces \cite{d-z}).

\subsection{Proof of Theorem \ref{thm:h1 and h2 pos beta intro}}

Set 
\[
E_{N}(x_{1},..,x_{N}):=H_{N}(x_{1},...,x_{N})/N
\]
so that the\emph{ mean energy} \ref{eq:def of mean free energy} can
be written as 
\[
E_{N}(\mu_{N}):=\int_{X^{N}}E_{N}\mu_{N}
\]
We denote by $F_{\beta_{N}}^{(N)}$ the corresponding \emph{mean free
energy functional} on $\mathcal{P}(X^{N})^{S_{N}}$, at inverse temperature
$\beta_{N}:$
\begin{equation}
F_{\beta_{N}}^{(N)}(\mu_{N}):=E^{(N)}(\mu_{N})+\frac{1}{\beta_{N}}D^{(N)}(\mu_{N})\label{eq:def of mean free energy}
\end{equation}
Now set $Y:=\mathcal{P}(X).$ By embedding $\mathcal{P}(X^{N}/S_{N})$
isometrically into $\mathcal{P}(Y),$ using the push-ward map $(\delta_{N})_{*},$
we can and identify the mean free energies $F^{(N)}$ with functionals
on $\mathcal{P}(Y),$ extended by $\infty$ to all of $\mathcal{P}(Y).$
We will identity $Y$ with image in $\mathcal{P}(Y)$ under the embedding
$\mu\mapsto\delta_{\mu}.$ 

The starting point of the proof of the LDP is the following reformulation
of Bryc's lemma in terms of Legendre-Fenchel transform, using the
Gibbs variational principle:
\begin{lem}
\label{lem:(Bryc+Gibbs):-Suppose-that}(Bryc+Gibbs): Suppose that
the Legendre-Fenchel transforms $f_{N}$ of the free energy functionals
$F_{\beta_{N}}^{(N)}$ converge point-wise to a function $f$ on $C^{0}(Y):$
\[
\lim_{N\rightarrow\infty}f_{N}(\Phi)=f(\Phi),
\]
 i.e. 
\[
\lim_{N\rightarrow\infty\Gamma\in}\inf_{\mathcal{P}(Y)}\left(F_{\beta_{N}}^{(N)}(\Gamma)+\left\langle \Phi,\Gamma\right\rangle \right)=-f(-\Phi),
\]
Then the LDP holds with speed $N\beta_{N}$ and rate functional 
\[
I(\mu):=f^{*}(\delta_{\mu}),
\]
 where $f^{*}$ is the Legendre-Fenchel transform of $f.$\end{lem}
\begin{proof}
Gibbs variational principle says that if $\mu_{\beta_{N}}^{(N)}$
is a well-defined probability measure, then 
\begin{equation}
\inf_{X^{N}}F^{(N)}=F^{(N)}(\mu_{\beta_{N}}^{(N)})=-\frac{1}{N\beta_{N}}\log\int_{X^{N}}e^{-\beta_{N}NE^{(N)}}\mu_{0}^{\otimes N}\label{eq:Gibbs}
\end{equation}
Indeed, rewriting $F^{(N)}(\mu_{N})=\frac{1}{\beta_{N}}D(\mu_{N},\mu_{\beta_{N}}^{(N)})-\frac{1}{N\beta_{N}}\log\int_{X^{N}}e^{-\beta_{N}NE^{(N)}}\mu_{0}^{\otimes N}$
this follows immediately from the fact that $D\geq0$ (which in turn
follows from Jensen's inequality). Hence, replacing $H^{(N)}$ with
the new Hamiltonian $H^{(N)}+N\delta_{N}^{*}(\Phi)$ and applying
Bryc's lemma concludes the proof.\end{proof}
\begin{rem}
Varadhan's lemma implies that the converse of the previous lemma also
holds. 
\end{rem}
In order to verify the criterion in the previous lemma we will use
the following lemma:
\begin{lem}
\label{lem:gamma relative}Under the hypotheses H1 and H2 and $\beta\in]0,\infty[$
the mean free energies $F_{\beta_{N}}^{(N)}$ Gamma-converge to the
lower semi-continuous linear functional $F_{\beta}(\Gamma)$ on $\mathcal{P}(Y),$
relative to $Y,$ where $F_{\beta}(\mu)$ is the macroscopic free
energy on $Y:$ 
\[
F_{\beta}:=E+D/\beta
\]
If moreover H3 holds, then the corresponding result also holds when
$\beta=\infty$\end{lem}
\begin{proof}
First assume that $\beta<\infty.$ The lower bound follows directly
from hypotheses H1 and H2 together with the fact that the mean entropy
functionals satisfy the lower bound in the Gamma-convergence (by subadditivity
\cite{ro-r}; see also Theorem 5.5 in \cite{h-m} for generalizations).
To prove the existence of recovery sequences we fix an element $\Gamma$
of the form $\delta_{\mu}$ and take the recovery sequence to be for
the form $(\delta_{N})_{*}\mu^{\otimes N}.$ Then the required convergence
follows from H1 together with the product property \ref{eq:entropi for product}.
Finally, when $\beta=\infty$ the previous argument for the existence
of a recovery sequence still applies as long as $\mu$ satisfies $D(\mu)<\infty.$
The general case then follows by a simple diagonal approximation argument
using H3. 
\end{proof}
Now, since the limiting functional $F(\Gamma)$ is affine and lower
semi-continuous (by Lemma \ref{lem:gamma is lsc}) and the set $Y$
is extremal in $\mathcal{P}(Y)$ the infimum of $F$ on $\mathcal{P}(Y)$
is attained in $Y$ (for example, by Choquet's theorem). Fixing a
continuous function $\Phi$ on $C^{0}(Y)$ and replacing $H^{(N)}$
with the new Hamiltonian $H^{(N)}+N\delta_{N}^{*}(\Phi),$ Lemma \ref{lem:conv of inf}
thus shows that the criterion in Lemma \ref{lem:(Bryc+Gibbs):-Suppose-that}
is satisfies. Hence the LDP holds with lower semi-continuous rate
functional $I(\mu)=f^{*}(\delta_{\mu}).$ Finally, extending $I$
to $\mathcal{P}(Y)$ by linearity this means that $I(\Gamma)$ is
the Legendre-Fenchel transform of $f,$ i.e. $I=f^{*}.$ But in our
case $f$ is itself defined as $f:=F^{*}$ and hence, $I=F^{**}=F$
since $F$ is convex (and even affine) and lower semi-continuous. 
\begin{rem}
\label{rem:H3'}An inspection of the proof of Theorem \ref{thm:h1 and h2 pos beta intro}
above reveals that, in the case $\beta=\infty,$ the hypothesis H3
may be replaced by the following weaker one:\end{rem}
\begin{itemize}
\item (H3)' The functional $F_{\beta}$ Gamma-converges towards $E,$ as
$\beta\rightarrow\infty$ 
\end{itemize}

\subsection{Proof of Corollary \ref{cor:cor mean field pos beta intro}}

First note that H1 is trivially satisfied (by the Fubini-Tonelli theorem).
To verify the second hypothesis H2 we may as well, by linearity, assume
that $M=m$ and that there is just one term with $W:=W_{m}(x_{1},...,x_{m}).$
Since $W$ is lower semi-continuous there exists a sequence of continuous
functions $W_{R}$ increasing to $W$ as $R\rightarrow\infty$ and
we denote by $E_{W_{R}}$ the corresponding functionals on $\mathcal{P}(X).$
It follows readily from the definitions that for any fixed $R>0$
\[
E_{W_{R}}(\delta_{N}(x_{1},...,x_{N}))+O(\frac{C_{R}}{N})=E_{W_{R}}^{(N)}(x_{1},x_{2},....,x_{N})
\]
and in particular
\[
E_{W}^{(N)}(\mu_{N})/N\geq\int E_{W_{R}}(\delta_{N}(x_{1},...,x_{N}))\mu_{N}+C_{R}/N
\]
But since $E_{W_{R}}$is continuous 
\[
\int E_{W_{R}}(\delta_{N}(x_{1},...,x_{N}))\mu_{N}=\int_{\mathcal{P}(X)}E_{W_{R}}(\mu)(\delta_{N})_{*}\mu_{N}\rightarrow\int_{\mathcal{P}(X)}E_{W_{R}}(\mu)\Gamma
\]
Hence, 
\[
\liminf_{N\rightarrow\infty}\int_{X^{N}}E^{(N)}\mu^{(N)}\geq\int E_{W_{R},V_{R}}(\mu)\Gamma
\]
for any $R>0.$ Finally, letting $R\rightarrow\infty$ and using the
monotone convergence theorem of integration theory concludes the proof.

\subsection{Proof of Theorem \ref{thm:h1 and h4 neg beta intro}}

First observe that if the LDP in the second point of the theorem holds
then integrating over all of $\mathcal{P}(X)$ reveals that the first
point holds. To prove the converse we fix $\beta>\beta_{0}$ and note
that Gibbs variational principle applied at the inverse temperature
$\beta-\epsilon$ gives 
\[
(\beta-\epsilon)F_{\beta-\epsilon}^{(N)}\geq-C_{\epsilon},
\]
 which we rewrite as 
\[
\beta F_{\beta}^{(N)}\geq\epsilon E^{(N)}-C_{\epsilon},
\]
Thus, by Gibbs variational principle, 
\[
\beta F_{\beta}^{(N)}(\nu^{\otimes N})\geq\beta F_{\beta}^{(N)}(\mu_{\beta}^{(N)})\geq\epsilon E^{(N)}(\mu_{\beta}^{(N)})-C_{\epsilon},
\]
for any fixed $\nu\in\mathcal{P}(Y).$ In particular, taking $\nu=\mu_{0}$
and using the hypothesis H1 gives that 
\[
E^{(N)}(\mu_{\beta}^{(N)})\leq C'
\]
But then the previous inequalities force 
\[
\mathcal{D}^{(N)}(\mu_{\beta}^{(N)})\leq C'
\]
Hence, by the hypothesis H4 
\[
\lim_{N\rightarrow\infty}E^{(N)}(\mu_{\beta}^{(N)})=\int_{\mathcal{P}(X)}E(\mu)\Gamma(\mu)
\]
 for any limit point $\Gamma$ of the laws of $\delta_{N}.$ As a
consequence we deduce precisely as before that the desired asymptotics
for the $\beta_{N}F^{(N)}$ hold. Finally, repeating the argument
with $H^{(N)}$ replaced by the new Hamiltonian $H^{(N)}+N\delta_{N}^{*}(\Phi)$
(which satisfies the same hypothesis) concludes the proof, just as
before. 
\begin{rem}
\label{rem:only concentration}If one only wants to prove that the
laws of $\delta_{N}$ concentrate on the minima of $F_{\beta}$ (rather
than proving a LDP) it is enough to show that the convergence of the
free energies hold for $\Phi=0$ (as in the original approach in \cite{m-s}).
As revealed by the previous proof this only requires that the hypothesis
H4 holds for the particular sequence $\mu_{\beta}^{(N)}.$ 
\end{rem}

\subsection{Proof of Corollary \ref{cor:mean field neg beta intro}}

Let us first show that the bound in the first point of Theorem \ref{thm:h1 and h4 neg beta intro}
is satisfied. We rewrite $-\beta H^{(N)}=\frac{1}{N}\sum_{i=1}^{N}f_{i},$
where $f_{i}$ is the sum of $W(x_{i},x_{j})$ over all $j$ such
that $j\neq i.$ The arithmetric-geometric means inequality gives
\[
\int_{X^{N}}e^{-\beta H^{(N)}}\mu_{0}^{\otimes N}\leq\sum_{i=1}^{N}\frac{1}{N}\int_{X^{N}}e^{f_{i}}\mu_{0}^{\otimes N}
\]
Performing the latter integral first over the $N-1$ variables different
from $x_{i}$ gives $C^{N-1}$where $C=\int e^{-\beta W(x,y)}\mu_{0}^{\otimes2},$
which is assumed finite, giving the desired bound.

Hence, by Theorem \ref{thm:h1 and h4 neg beta intro}, it will be
enough to show that H4 is satisfied (for any sequence $\mu_{N})$.
To this end we will apply a duality argument. First recall that given
a measure space $(\mathcal{X},\mu)$ and a finite Young function $\theta$
on $\R$ (i.e. a non-negative even lower semi-continuous convex function)
the corresponding \emph{large Orlitz space }is defined by 
\[
L_{\theta}(\mathcal{X},\mu):=\left\{ f:\,\,\exists\alpha>0:\,\int\theta(\alpha f)\mu<\infty\right\} 
\]
and the corresponding \emph{small Orlitz space }is defined by 
\[
M_{\theta}(\mathcal{X},\mu):=\left\{ f:\,\,\forall\alpha>0:\,\int\theta(\alpha f)\mu<\infty\right\} 
\]
(where all functions $f$ are assumed measurable). The space $L_{\theta}(\mathcal{X},\mu)$
may be equipped with a norm $\left\Vert \cdot\right\Vert _{\theta},$
called the Luxemburg norm, which turns $L_{\theta}(\mathcal{X},\mu)$
and its subspace $M_{\theta}(\mathcal{X},\mu)$ into Banach spaces:
\[
\left\Vert f\right\Vert _{\theta}:=\inf\left\{ b>0:\,\int\int\theta(b^{-1}f)\mu\leq1\right\} ,
\]
 i.e. the gauge of the set (unit-ball) 
\[
\left\{ f:\,\int\int\theta(f)\mu\leq1\right\} ,
\]
By the Hölder-Young inequality 
\[
\left|\int fg\mu\right|\leq2\left\Vert f\right\Vert _{\theta}\left\Vert g\right\Vert _{\theta^{*}},
\]
where $\theta^{*}$ is the Young function defined as the Legendre-Fenchel
transform of $\theta.$ In particular, for any $g\in L_{\theta^{*}}$
$f\mapsto\int fg\mu$ defines a continuous function on $L_{\theta}$
with bounded operator norm, i.e. $L_{\theta^{*}}\subset L_{\theta}^{*},$
where $L^{*}$ denotes the Banach space dual of a Banach space $L,$
endowed with the operator norm. To apply this in the present context
we note that 
\[
E^{(N)}(\mu_{\beta}^{(N)})=\int_{X^{2}}W\rho_{N}\mu_{0}^{\otimes2}
\]
 where $\rho_{N}$ is the density of the second marginal of $\mu_{\beta}^{(N)}.$
The assumption that $D^{(N)}(\mu_{\beta}^{(N)})\leq C$ implies that
\[
\int_{X^{m}}(\rho_{N}\log\rho_{N})\mu_{0}^{\otimes m}\leq C2,
\]
 according to \ref{eq:ineq for entropy}.

Now set $\theta(s):=e^{s}-s-1.$ Then $\theta^{*}(s)=(s+1)\log(1+s)-s.$
By the previous entropy inequality for $\rho_{N}$ the sequence $\{\rho_{N}\}$
stays in a fixed ball in $L_{\theta^{*}}$ and hence, by the Hölder-Young
inequality, $\{\rho_{N}\}$ stays in a fixed ball in the dual Banach
space $L_{\theta}^{*}.$ By weak compactness it then follows that
there exists $\Lambda\in L_{\theta}^{*}$ such that for any $g\in L_{\theta}$
\[
\int_{X^{m}}\rho_{N}g\mu_{0}^{\otimes m}\rightarrow\left\langle \Lambda,g\right\rangle 
\]
(after perhaps passing to a subsequence). Now, since $\rho_{N}g\mu$
is a probability measure we may also assume that there exists $\rho\in L_{\theta}$
such that 
\[
\int_{X^{m}}\rho_{N}u\mu_{0}^{\otimes m}\rightarrow\int_{X^{m}}\rho u\mu_{0}^{\otimes m}
\]
 for any continuous function $u.$ In our case $g=W$ and we just
need to check that $\left\langle \Lambda,g\right\rangle =\left\langle \Lambda,\rho\right\rangle .$
But, by assumption, $W\in M_{\theta}$ and by the general duality
theorems in \cite{r-r,le} the topological dual of $M_{\theta}$ identifies
with $L_{\theta^{*}},$ i.e. any continuous functional $\Lambda$
on $M_{\theta}$ is obtained by integrating against a (unique) $\rho\in L_{\theta^{*}},$
which concludes the proof.

\subsection{Proof of Theorem \ref{thm:negive beta uniform integr}}

Given a compact metric space $X$ we endow $Y(:=\mathcal{P}(X))$
with the Wasserstein $L^{1}-$metric $d,$ which is compatible with
the weak topology:
\[
d(\mu,\nu)=\sup_{f:\,L(f)\leq1}\int f(\mu-\nu),
\]
 where $f$ is Lipschitz continuous on $X$ with Lipschitz constant
$1.$ Since $\int(\mu-\nu)=0$ we may as well assume that $f(x_{0})=0$
for a fixed point $x_{0}$ and hence that $|f(x)|\leq C_{X}$ where
$C_{X}$ is independent of $f$ (since $X$ is compact and, in particular,
has bounded diameter). 

We fix, as before, a continuous function $\Phi$ on $Y:=\mathcal{P}(X).$
Without loss of generality we may as well assume that $W,\Phi\geq0.$ 

First observe that when $\beta>\beta_{0}$ we have 
\begin{equation}
Z_{N,\beta}[\Phi]:=\int_{X^{N}}e^{-\beta_{N}\left(H^{(N)}+N\delta_{N}^{*}(\Phi)\right)}\mu_{0}^{\otimes N}\leq C_{\beta}^{N},\label{eq:bound on part function in proof ldp vortex}
\end{equation}
 precisely as in the beginning of the proof of Cor \ref{cor:mean field neg beta intro}
(using that $\Phi$ is bounded). Using the convergence of the mean
energies and Gibbs variational principle, as before, we thus have
\begin{equation}
-C_{\beta}\leq\limsup_{N\rightarrow\infty}-\frac{1}{N}\log Z_{N,\beta}[\Phi]\leq\inf_{\mu\in\mathcal{P}(X)}\left(\beta E(\mu)+\Phi(\mu)+D(\mu)\right)\label{eq:lower bound on free}
\end{equation}

\begin{lem}
\label{lem:free energy is lsc for beta neg}For any $\beta>\beta_{0}$
the corresponding (scaled) free energy functional $\beta F_{\beta}$
on $\mathcal{P}(X)$ is lower semi-continuous.\end{lem}
\begin{proof}
First observe that by the inequality \ref{eq:lower bound on free}
(applied to $\Phi=0):$ 
\[
\beta F_{\beta}\geq-C_{\beta}.
\]
 Now, applying the previous bound to $\beta-\epsilon>\beta_{0}$ reveals
that
\[
\beta F_{\beta}\geq\epsilon E-C'.
\]
Given $\mu$ in $\mathcal{P}(X)$ with $E(\mu)<\infty$ we set $u_{\mu}(x):=\int W(x,y)\mu.$
Then 
\[
E(\mu)=\int_{X}u_{\mu}\mu
\]
Using the previous estimate it will, to prove the lemma, be enough
to verify the following ``macroscopic'' version of H4 for any sequence
$\mu_{j}$ converging weakly towards $\mu:$ 
\[
D(\mu_{j})\leq C\implies E(\mu_{j})\rightarrow E(\mu)
\]
(compare \cite[Theorem 2.17]{bbgez}). To this end we set $u_{j}:=u_{\mu_{j}}$
and observe that the Hölder-Young inequality (with $\theta$ of exponential
type, as in the proof of Cor \ref{cor:mean field neg beta intro})
implies that $|u_{j}|\leq C$ (using the assumption on $W).$ Hence,
using the Hölder-Young inequality again it will be enough to show
that, $\left\Vert u_{j}-u\right\Vert _{\theta}\rightarrow0,$ or equivalently
that, for any given $a>0$ 
\[
\int_{X}\theta(a(u_{j}-u))\mu_{0}\rightarrow0.
\]
Since $\theta(t)\leq te^{t}$ and $|u_{j}|\leq C$ it will thus be
enough to show that 
\[
\left\Vert u_{j}-u\right\Vert _{L^{1}(X,\mu_{0})}\rightarrow0.
\]
 By the lower semi-continuity of $W,$ 
\[
\liminf_{j\rightarrow\infty}u_{j}\geq u
\]
 the desired convergence will follow from general measure theory if
$\int u_{j}\mu_{0}\rightarrow\int u\mu_{0}$ (using that $u_{\mu}\geq0$
if $W$ is normalized so that $W\geq0).$ But 
\[
\int_{X}u_{j}\mu_{0}=\int_{X}v\mu_{j}(y),\,\,\,v:=\int_{X}W(x,y)\mu_{0},
\]
 where $v$ is bounded, by the previous argument (since $\mu_{0}$
trivially has finite entropy). In particular, $v$ is in the little
Orlitz space $M_{\theta}(X,\mu_{0})$ and since $D(\mu_{j})\leq C$
the desired convergence then follows from the duality argument towards
the end of the proof of Cor \ref{cor:mean field neg beta intro}. 
\end{proof}
Now, to prove the LDP we need, in view of Lemma \ref{lem:(Bryc+Gibbs):-Suppose-that},
to complement the upper bound on $-\frac{1}{N}\log Z_{N,\beta}[\Phi]$
in formula \ref{eq:lower bound on free} with a corresponding lower
bound. To this end it seems natural to try to extend the Orlitz space
duality argument in the proof of Cor \ref{cor:mean field neg beta intro}
to the present setting, exploiting the uniform bound on the entropy
of the marginals. But here we will instead take another road (inspired
by \cite{clmp,k2}), exploiting the stronger $L^{p}-$bounds provided
by the following lemma.
\begin{lem}
\label{lem:bound on marginals}Let $\Phi$ be a given Lipschitz continuous
function on $Y:=\mathcal{P}(X)$ and fix $\beta>\beta_{0.}$ Then
the following estimate holds for the densities $\rho_{j}^{(N)}$ of
the $j$ th marginal of the Gibbs measures corresponding to the Hamiltonian
$H^{(N)}+N\delta_{N}^{*}(\Phi):$ 
\[
\rho_{j}^{(N)}(x_{1},..x_{j})\leq C_{j}e^{-\frac{1}{N}\sum\sum_{k\neq l\leq j}W(x_{k},x_{l})}
\]
as $N\rightarrow\infty.$ In particular, for any $p>1$ $\rho_{j}^{(N)}(x_{1},..x_{j})$
is uniformly bounded in $L^{p}$ as $N\rightarrow\infty.$\end{lem}
\begin{proof}
To fix ideas we start with the case $\Phi=0,$ following closely the
proof of Theorem 3.1 in \cite{clmp}. Set 
\[
W(X,Y):=\sum_{x\in X,y\in Y}W(x_{i},y_{j}),\,\,\,\,d\mu(Y):=\mu_{0}^{\otimes N-1}
\]
where $X:=\{x_{1},..,x_{j}\}$ and $Y:=\{x_{j+1},...,x_{N}\}.$Then
we can decompose $E^{(N)}(x_{1},...,x_{N})=\frac{1}{N}W(X,X)+\frac{1}{N}W(X,Y)+\frac{1}{N}W(Y,Y).$
Accordingly, 
\[
\rho_{j}^{(N)}(X)=e^{-\frac{\beta}{N}W(X,X)}\frac{1}{Z_{N}}\int e^{-\frac{\beta}{N}W(X,Y)}e^{-\frac{\beta}{N}W(Y,Y)}d\mu(Y)
\]
Applying Hölder's inequality with $p=N$ (and thus $q=1+1/N-1))$
the integral in the right hand side is bounded from above by 
\[
=\left(\int e^{-\beta W(X,Y)}d\mu(Y)\right)^{1/N}\left(\int e^{-\frac{\beta}{N}qW(Y,Y)}d\mu(Y)\right)^{1/q}
\]
By assumption the integral appearing in the first factor above is
bounded from above by a $A^{N}.$ It will thus be enough to show that
the second integral is controlled by $Z_{N}$ in the sense that it
is bounded from above by a uniform constant times $Z_{N}.$ To this
end we will apply Hölder's inequality again, now with conjugate exponents
$u$ and $w$ with $u$ sufficiently close to $1$ (to be quantified
below). We thus rewrite
\[
q=\frac{1}{u}+\left(q-\frac{1}{u}\right)
\]
and apply Hölder's inequality. Since $w(q-1/u)=1+w(q-1)=1+\frac{w}{(N-1)}$
this gives 
\begin{equation}
\int e^{-\frac{\beta}{N}qW(Y,Y)}d\mu(Y)\leq\left(\int e^{-\frac{\beta}{N}W(Y,Y)}d\mu(Y)\right)^{1/u}\left(\int e^{-\frac{\beta}{N}(1+\frac{w}{(N-1)})W(Y,Y)}d\mu(Y)\right)^{1/w}\label{eq:second h=0000F6lder}
\end{equation}
Hence, taking $w=\epsilon(N-1)$ for a sufficiently small positive
number $\epsilon$ the first factor is controlled by $Z_{N}$ (since
$W\geq0)$ and the integral in second factor is controlled by $Z_{N.(1+\epsilon)\beta}\leq B^{N}$
(by \ref{eq:bound on part function in proof ldp vortex}). Since $w$
is of the order $N$ this concludes the proof when $\Phi=0.$ To treat
the general case we will use the following
\[
\mbox{Claim:\,}\left|\Phi(\frac{1}{N}\sum_{i=1}^{N}\delta_{x_{i}})-\Phi(\frac{1}{N-j}\sum_{i=j+1}^{N}\delta_{x_{i}})\right|\leq C\frac{1}{N},
\]
Accepting the claim for the moment and introducing the notation $\Phi(\frac{1}{N-j}\sum_{i=j+1}^{N}\delta_{x_{i}})=\phi(Y)$
we have 
\[
\rho_{j}^{(N)}(X)\leq\frac{e^{C}}{Z_{N}}e^{-\frac{\beta}{N}W(X,X)}\int e^{-\frac{\beta}{N}W(X,Y)}e^{-\frac{\beta}{N}(W(Y,Y)+N\phi(Y))}d\mu(Y)
\]
We then use first Hölder's inequality with $p$ and $q$ and then
with $u$ and $v$ exactly as above to get the same factors as above
apart from the last factor in formula \ref{eq:second h=0000F6lder}
which now becomes 
\[
\int e^{-\frac{\beta}{N}(1+\gamma)(W(Y,Y)+N\Phi(y))}d\mu(Y)
\]
 which is bounded from above by $C'^{N},$ according to the estimate
\ref{eq:bound on part function in proof ldp vortex} (when $\gamma$
is sufficiently small). This proves the lemma once we have verified
the claim above. To this end we assume to simplify the notation that
$j=1$ (the general case is similar) and observe that setting $\mu:=\frac{1}{N}\sum_{i=1}^{N}\delta_{x_{i}}$
and $\nu:=\frac{1}{N-1}\sum_{i=2}^{N}\delta_{x_{i}}$ gives 
\[
N(\mu-\nu)=\delta_{x_{1}}-\frac{1}{N-1}\sum_{i=2}^{N}\delta_{x_{i}}
\]
Hence, for any $f$ such that $|f|\leq C_{X}$ we have 
\[
|\int f(\mu-\nu)|\leq\frac{1}{N}\left(C_{X}+\frac{1}{N-1}\sum_{i=2}^{N}C_{X}\right)\leq\frac{1}{N}2C_{X},
\]
and hence $d(\mu,\nu)\leq2C_{X}/N.$ But then the claim follows directly
from the Lipschitz continuity of $\Phi.$ 
\end{proof}
Now, to verify the missing lower bound on 
\[
-\frac{1}{N}\log Z_{N,\beta}[\Phi]\left(=F_{\beta}^{(N)}(\mu_{\beta}^{(N)})\right)
\]
 we first claim that it will be enough to verify the case when $\Phi$
Lipschitz continuous. Indeed, any continuous function $\Phi$ on a
compact metric space $Y$ can be written as a uniform limit $\Phi^{(R)}$
of Lipschitz continuous function (for example, $\Phi^{(R)}(x):=\inf_{Y}\left(\Phi(y)+Rd(x,y)\right)$
increases to $\Phi,$ as $R\rightarrow\infty,$ and has Lipschitz
constant $R).$ Moreover we may, after relabeling the sequence, assume
that $|\Phi_{\epsilon}-\Phi|\leq\epsilon.$ But since $\Phi\mapsto\frac{1}{N\beta}\log Z_{N,\beta}[\Phi]$
is increasing and $\frac{1}{N\beta}\log Z_{N,\beta}[\Phi+c]=\frac{1}{N\beta}\log Z_{N,\beta}[\Phi]+c$
for any $c\in\R$ we get 
\[
\left|\frac{1}{N\beta}\log Z_{N,\beta}[\Phi]-\frac{1}{N\beta}Z_{N,\beta}[\Phi_{\epsilon}]\right|\leq\epsilon,
\]
which proves the claim. Next, we consider the sequence $\mu_{N}=\mu_{\beta_{N}}^{(N)}$
of Gibbs measures corresponding to the Hamiltonian $H^{(N)}+N\delta_{N}^{*}(\Phi),$
for $\Phi$ Lipschitz continuous and decompose 
\[
E^{(N)}(\mu_{\beta}^{(N)})=\int_{X^{2}}W\rho_{2}^{(N)}\mu_{0}^{\otimes2}+\left\langle \Phi,\Gamma_{N}\right\rangle 
\]
 By continuity the second term above converges towards $\left\langle \Phi,\Gamma\right\rangle .$
To prove the desired lower bound on $F_{\beta}^{(N)}(\mu_{\beta}^{(N)})$
it will thus, just as in the proof of Cor \ref{cor:mean field neg beta intro},
be enough to show that 
\begin{equation}
\lim_{N\rightarrow\infty}E^{(N)}(\mu_{\beta}^{(N)})=\int W\rho_{2}\mu_{0}^{\otimes2},\label{eq:limit of mean energy in proof ldp negative beta noll}
\end{equation}
 for any weak limit point $\rho_{2}\mu_{0}^{\otimes2}$ of $\rho_{2}^{(N)}\mu_{0}^{\otimes2}.$
To this end we recall that, by the previous lemma, $\rho_{2}^{(N)}$
is uniformly bounded in $L^{p},$ as $N\rightarrow\infty,$ for any
fixed $p>1.$ Hence, by standard $L^{p}-$duality the limit \ref{eq:limit of mean energy in proof ldp negative beta noll},
follows from the fact that $W\in L^{q}$ for some (any) $q>1,$ since
by assumption, $e^{\epsilon W}\in L^{1}$ for any sufficiently small
positive number $\epsilon.$

\subsection{\label{sub:Relations-to-Gamma in text}Relations to Gamma-convergence
of $E^{(N)}$ on $\mathcal{P}(X):$ Proof of Cor \ref{cor:h1 h2 h3 gives gamma intro}}

First observe that the required lower bound on $E^{(N)}(x^{(N)})$
is obtained by taking $\mu_{N}$ to be  the normalized $S_{N}-$orbit
in $X^{N}$of the Dirac measure supported at $x^{(N)}.$ Then 
\[
E^{(N)}(x^{(N)})=E^{(N)}(\mu_{N})
\]
and since $\mu_{N}$ converges towards $\Gamma:=\delta_{\mu}$ this
proves the desired lower bound. 

Next, to construct recovery sequences we take a sequence $\beta_{N}\rightarrow\infty.$
By Theorem \ref{thm:h1 and h2 pos beta intro} an LDP holds with rate
functional $E(\mu).$ In particular, the lower bound in the LDP gives
that for any $\delta>0$ we have for $\epsilon\leq\epsilon_{\delta}$
\[
-E(\mu)-\delta\leq\limsup_{N\rightarrow\infty}\frac{1}{N\beta_{N}}\log\int_{B_{\epsilon}(\mu)}e^{-\beta_{N}H^{(N)}}\mu_{0}^{\otimes N}\leq
\]
\[
\leq\limsup_{N\rightarrow\infty}\left(-\inf_{B_{\epsilon}(\mu)}E^{(N)})+\frac{1}{\beta_{N}N}\log\int_{B_{\epsilon}(\mu)}\mu_{0}^{\otimes N}\right)
\]
Hence, if $D(\mu)<\infty$ then Sanov's theorem (i.e. the LDP for
$H^{(N)}=0)$ shows that there exists a sequence $x_{\epsilon}^{(N)}\in B_{\epsilon}(\mu)$
such that 
\[
-E(\mu)-\delta\leq\limsup_{N\rightarrow\infty}\left(-E^{(N)}(x_{\epsilon}^{(N)})\right)
\]
Now a diagonal argument shows that any such $\mu$ admits a recovery
sequence. Finally, the regularity assumption allows us to deduce the
existence of a recovery sequence for any $\mu.$

\section{Concl\label{sec:Concluding-remarks}uding remarks}

\subsection{A weaker form of the hypothesis H2}

Let us come back to the setting of Theorem \ref{thm:h1 and h2 pos beta intro}
and observe that the hypothesis H2 may be replaced by the following
one, which is a priori weaker (see the beginning of Section \ref{sub:Relations-to-Gamma in text}): 
\begin{itemize}
\item (H2') For any sequence of $x^{(N)}\in X^{N}$ such that $\delta_{N}(x^{(N)})\rightarrow\mu$
weakly in $\mathcal{P}(X)$ we have 
\[
\liminf_{N\rightarrow\infty}\frac{1}{N}H^{(N)}(x^{(N)})\geq E(\mu),
\]
 
\end{itemize}
In other words, (H2') says that $E^{(N)}:=H^{(N)}/N,$ when viewed
as a functional on $\mathcal{P}(X),$ satisfies the lower bound property
which is one of the two requirements for the Gamma-convergence of
$E^{(N)}$ towards $E(\mu),$ where $E(\mu)$ denotes, as before,
the macroscopic mean energy whose existence is postulated in hypotheses
H1. 
\begin{thm}
\label{thm:H2'}The conclusion of Theorem \ref{thm:h1 and h2 pos beta intro}
remains valid if H2 is replaced by H2'\end{thm}
\begin{proof}
Just as in the proof of Theorem\ref{thm:h1 and h2 pos beta intro},
in order to verify the convergence in Bryc's lemma, we may without
loss of generality assume that $\Phi=0.$ Moreover, exactly as before
H1 combined with the Gibbs variational principle yields the upper
bound on $-\log Z_{N,\beta_{N}}.$ To prove the lower bound first
note that for any fixed $\mu\in\mathcal{P}(X)$ and $\epsilon>0$
Sanov's theorem gives, just as in the proof of Cor \ref{cor:h1 h2 h3 gives gamma intro}
in Section \ref{sub:Relations-to-Gamma in text}, that \emph{
\begin{equation}
\limsup_{\epsilon\rightarrow\infty}\limsup_{N\rightarrow\infty}\frac{1}{N\beta_{N}}\log\int_{B_{\epsilon}(\mu)}e^{-\beta_{N}H^{(N)}}\mu_{0}^{\otimes N}\leq-\liminf_{\epsilon\rightarrow\infty}\liminf_{N\rightarrow\infty}\inf_{B_{\epsilon}(\mu)}E^{(N)})-\frac{1}{\beta}D(\mu),\label{eq:proof of theorem mean plus gamma}
\end{equation}
} where, by hypothesis H2', the right hand side is bounded from above
by $-E(\mu)-\frac{1}{\beta}D(\mu):=-F_{\beta}(\mu).$ Now, for any
fixed $\delta>0$ we cover the compact space $\mathcal{P}(X)$ by
a finite number $M_{\delta}$ of balls $B_{\delta}(\mu_{\delta,i})$
with $i=1,..,M_{\delta}.$ Then \emph{
\[
\limsup_{N\rightarrow\infty}\frac{1}{N\beta_{N}}\log\int_{X^{N}}e^{-\beta_{N}H^{(N)}}\mu_{0}^{\otimes N}\leq0+\limsup_{N\rightarrow\infty}\frac{1}{N\beta_{N}}\log\int_{B_{\delta}(\mu_{\delta})}e^{-\beta_{N}H^{(N)}}\mu_{0}^{\otimes N},
\]
} where $\mu_{\delta}$ is the center of the ball with the largest
integral. Next, fix $\epsilon>0$ and denote by $\mu$ a limit point
in $\mathcal{P}(X)$ of the family $\mu_{\delta}.$ For any sufficiently
small $\delta_{j}$ we have $B_{\delta_{j}}(\mu_{\delta_{j}})\subset B_{\epsilon}(\mu).$
Hence, estimating the right hand side in the previous formula with
the integral over $B_{\epsilon}(\mu)$ and using the inequality \ref{eq:proof of theorem mean plus gamma}
gives \emph{
\[
\limsup_{\epsilon\rightarrow\infty}\limsup_{N\rightarrow\infty}\frac{1}{N\beta_{N}}\log\int_{X^{N}}e^{-\beta_{N}H^{(N)}}\mu_{0}^{\otimes N}\leq-F_{\beta}(\mu)\leq-\inf_{\mathcal{P}(X)}F_{\beta},
\]
} which shows that Bryc's lemma can be applied, just as before, to
deduce the LDP in question. 
\end{proof}
A result essentially equivalent to the previous theorem appears in
\cite{bo}. Combining Theorem \ref{thm:H2'} and Cor \ref{cor:h1 h2 h3 gives gamma intro}
thus reveals that the hypotheses H1 and H2' actually implies the Gamma-convergence
of $\frac{1}{N}H^{(N)}$ towards $E$ on $\mathcal{P}(X),$ if the
approximation hypothesis H3 holds. 

Finally, let us point out that it seems unlikely that, in general,
the assumption that $\frac{1}{N}H^{(N)}$ Gamma-converges towards
a functional $E$ on $\mathcal{P}(X)$ is not enough to deduce a LDP
(even if one also assumes H3). On the other hand, as shown in \cite{berm8},
one does get an LDP for any $\beta\in]0,\infty]$ under an assumption
of quasi-superharmonicity:
\begin{thm}
\cite{berm8} \label{thm:g-conv of sh energies}Let $H^{(N)}$ be
a sequence of lower semi-continuous symmetric functions on $X^{N},$
where $X$ is a compact Riemannian manifold.\textup{ Assume that }
\begin{itemize}
\item \textup{The sequence }$\frac{1}{N}H^{(N)}$ on $X^{N}$ (identified
with a sequence of functions on $\mathcal{P}(X)$ Gamma-converges
towards a functional $E$ on $\mathcal{P}(X)$ 
\item $H^{(N)}$ is uniformly quasi-superharmonic, i.e. \emph{$\Delta_{x_{1}}H^{(N)}(x_{1},x_{2},...x_{N})\leq C$
on $X^{N}$}
\end{itemize}
Then, for any sequence of positive numbers $\beta_{N}\rightarrow\beta\in]0,\infty]$
the measures $\Gamma_{N}:=(\delta_{N})_{*}e^{-\beta_{N}H^{(N)}}$
on $\mathcal{M}_{1}(X)$ satisfy, as $N\rightarrow\infty,$ a LDP
with \emph{speed} $\beta_{N}N$ and good \emph{rate functional} 
\begin{equation}
F_{\beta}(\mu)=E(\mu)+\frac{1}{\beta}D_{dV}(\mu)\label{eq:free energy func theorem gibbs intro-1}
\end{equation}

\end{thm}
This is not hard to see when $\beta=\infty,$ but for $\beta<\infty$
the proof hinges on a submean inequality for quasi-subharmonic functions
with a distortion factor which is subexponential in the dimension,
proved in \cite{berm8}.

\end{document}